\documentclass[12pt]{article}
\usepackage{amsmath,amsthm,amssymb}
\usepackage{xcolor}

\allowdisplaybreaks
\newcommand\ftz{\footnotesize}
\newcommand\frct [2]
{\ensuremath{\raisebox{.3ex}{\ftz #1}\!/\!\raisebox{-0.3ex}{\ftz#2}}}
\newcommand\given{\, |\, }
\newcommand\almostsure{\buildrel a.s. \over  \longrightarrow}
\newcommand\inprob{\buildrel {\mathbb P} \over  \longrightarrow}
\newcommand\convD{\buildrel {\cal D} \over \longrightarrow }
\newcommand\Prob{{\mathbb P}}
\newcommand\field{{\mathbb F}}
\newcommand\E{{\mathbb E}}
\newcommand\V{{\mathbb V}{\rm ar}}
\newcommand\Cov{{\mathbb C}{\rm ov}}
\newcommand\matA{{\bf A}}
\newcommand\matL{{\bf L}}
\newcommand\matP{{\bf P}}
\newcommand\matT{{\bf T}}
\newcommand\matU{{\bf U}}
\newcommand\matW{{\bf W}}
\newtheorem{lemma}{Lemma}[section]
\newtheorem{cor}{Corollary}[section]
\newtheorem{theorem}{Theorem}[section]
\newtheorem{example}{Example}[section]

\usepackage{etoolbox}

\begin{document}
\begin{center}
{\bf \Huge Probabilistic analysis of arithmetic coding showing its robustness}

\bigskip
{\bf \large
Hosam M.\ Mahmoud}\footnote{Department of Statistics,
The George Washington University,
Washington, D.C. 20052, U.S.A. Email: hosam@gwu.edu} \qquad
{\bf \large Hans J.\ Rivertz}\footnote{Department of Computer Science, Faculty of Information Technology and Electrical Engineering, Norwegian University of Science and Technology.
Email: hans.j.rivertz@ntnu.no}

\bigskip
{\bf \large \today}
\end{center}

\bigskip 
\section*{Abstract}     
We probabilistically analyze the performance of the arithmetic coding algorithm under a probability model for binary data in which  a message is received by a coder from 
a source emitting independent equally distributed bits, with 1 occurring with probability 
$p\in(0,1)$ and 0 occurring with probability $1-p$. 

We establish a functional equation for the bivariate moment generating function for the two ends of the final interval delivered by the algorithm. Via the method of moments, 
we show that the transmitted message converges in distribution to the standard 
continuous uniform random variable on the interval [0,1]. It is remarkable that
the limiting distribution is the same for all $p$, indicating robustness in the performance  of arithmetic coding across an entire family of bit distributions. The nuance with $p$ appears only in the rate of convergence.   

\bigskip\noindent
{\bf Keywords:}  
            Random structure, 
            robust algorithm,
            coding,       
            combinatorial probability, 
            method of moments.
                        
\bigskip
\noindent{\bf AMS subject classifications:} 

\medskip
\hskip 1cm 
Primary: 	68P30,  %	Coding and information theory (compaction, compression, models of communication, encoding schemes, etc.) (aspects in computer science)
68P20.  	%Information storage and retrieval of data

\hskip 1cm 
   Secondary: 60C05,     % Combinatorial probability
      60F05.                   % Central limit and other weak theorems       
\section{Arithmetic coding}
Entropy coders are lossless compression-recovery
 algorithms that are applied in many data reduction schemes. Huffman coding and arithmetic coding are two such methods that differ significantly. Although both methods try to reach Shannon's fundamental lower limit of lossless compression~\cite{Shannon}, Huffman coding only achieves this in very special cases.

Arithmetic coding is an elegant algorithm that transmits an entire message (possibly a long one) as a single number (a codeword, possibly quite short).
 The method's attractiveness %of the method 
 stems from its unique decipherability to recover the transmitted 
 codeword. The transmission is achieved by two coupled stages. The message is fed into a coder, which operates on its letters transforming them into a codeword. The codeword is sent to a decoder which recovers the message in its original form.    

The concept of arithmetic coding was introduced by Jorma Rissanen in the mid-1970's. 
However, there are no noticeable publications on it until the mid-1980's because the initial publication
was an internal report within IBM and was not widely circulated outside that community. 
The two publications~\cite{Glen,Witten} polarized the algorithm.
The survey~\cite{Said} is widely cited. For a textbook style presentation, see Chapter 6 of~\cite{Hankerson}. For a general background on information theory, we refer the reader 
to~\cite{Cover} or~\cite{Sayood}.

All data are ultimately transmitted as binary strings, so we might as well deal with data in binary form.
Let $n$ be the length (in bits) of a message received by the arithmetic coder. 

Our purpose is to study the probabilistic behavior of the algorithm over a large
number of {\em random} messages to assess an average behavior over a family of bit distributions. The analysis reveals surprising 
robustness---there is not much change in
the limiting behavior across the probability model considered. The nuance comes only
in the rate of convergence.

The message is composed bit by bit. Let us call the $n$th bit received $B_n$.
As a model of randomness we take 
the bits of the message to be independent, and for any $n$, we have
$$\Prob(B_n=1) = p, \qquad\Prob(B_n =0) = 1 - p =: q.$$ 
We  assume the 
algorithm has direct access to the $i$th bit 
$B_i$ of 
the message, for $i=1, \ldots,n$.  
%Hosam: No need for M any more.
For example, with $n=8$, the received message might be $10111101$. By the independence of the bits,
the probability of this string is $p^6 q^2$. 
Initially, before processing any bits, the algorithm takes the codeword to be in the interval
$[X,Y] = [0,1]$. Later $X$ and $Y$ coalesce toward each other, and the length of the interval is $Y-X$. 
The appearance of a 0 as the next bit shrinks the interval from the upper end by
setting $Y=X + q(Y-X)$, and if the next bit is 1, the algorithm  raises $X$ to position
$X + q(X-Y)$.
After all $n$ bits have been processed, we have an interval $[X,Y]$, which one would
expect to be small if the message is long. 
The decoding stage of the algorithm can decipher the message uniquely given any number in the final interval. It is therefore advantageous to use a short codeword, a number in the interval with the fewest possible significant bits via some extra algorithmic steps to achieve the best possible compression rate.
In  some implementations, the codeword is taken to be the middle point $\frct 1 2(X+Y)$.

We present the algorithm in pseudocode in the probabilistic setting discussed: 
 
\bigskip
         {\phantom {X}}\\
      {\phantom {X}} \hskip 1cm $X\gets 0$\\
      {\phantom {X}} \hskip 1cm $Y\gets 1$ \\
      {\phantom {X}} \hskip 1cm {\bf for} $i$ {\bf from} 1 {\bf to} $n$ {\bf do}\\
         {\phantom {X}}  \hskip 2cm{\bf  if} $%B
         B_i = 0$  \\
         {\phantom {X}}          \hskip 3cm {\bf then}  $Y \gets X + q(Y-X)$\\
         {\phantom {X}}        \hskip 3cm {\bf else}  $X \gets X + q(Y-X)$
\begin{example}         
Suppose the message to transmit is 011101, and $p=\frct 13$.  We have the following
sequence of truncations
of the initial interval 
\begin{align*}
0 &\implies \Big[0, \frac 1 3\Big] , \qquad{\rm of\ length}\ \frac 13\\
1 &\implies  \Big[\frac1 9, \frac 1 3\Big] ,\qquad {\rm of\ length}\ \frac 2 9\\
1 &\implies   \Big[\frac 5{27}, \frac 1 3\Big] , \qquad {\rm of\ length} \ \frac 4 {27}\\
1 &\implies   \Big[\frac {19} {81}, \frac 1 3\Big], \qquad {\rm of\ length}\ \frac 8 {81}\\
0 &\implies   \Big[\frac {19} {81}, \frac {65} {243}\Big], \qquad {\rm of\ length}\ \frac 8 {243}\\
1 &\implies   \Big[\frac {179}{729}, \frac {65} {243}\Big], \qquad {\rm of\ length}\ \frac {16}{729}.
\end{align*}

The coder delivers the interval $[\frac {179}{729}, \frac {65} {243}]$, which in binary is
$$[\,0.0011111011\ldots\,\,, 0.0100010010\ldots\,].$$
We can choose any number within this interval as a value for the codeword.
We take the shortest prefix on which the binary representation of the two ends disagree as
the net result,
which in this example is 0.01. A message of length 6 bits is
successfully transmitted as a codeword in two bits, achieving good compression. 

An implementation choosing to transmit the middle point goes for
$$0.01000001101\ldots$$ 
as the codeword, which is infinite. The shortest finite prefix that remains within the final version is taken as the codeword, which in this example is 0.01.
This short codeword can only be used when the decoder knows in advance that there are only two bits.
In practice, multiple encoded messages might appear sequentially in a bit stream.
This requires a subinterval within the code interval.
For the given encoding, the subinterval is $[0.010000,0.010001]$, which is uniquely represented by 0.010000. This is not a compression, which is as expected, as the maximal compression rate with $p=1/3$ is $-\frac 1 3\log_2\frac 1 3 - \frac 2 3 \log_2\frac 2 3\approx 0.92$~\cite{Shannon}.
\end{example}
\section{Notation}
The standard probability notation $\convD$,
$\inprob$ and $\almostsure$ respectively stand for convergences in distribution,
in probability and almost surely.     

For an $r$-by-$r$ matrix, its $r$ eigenvalues $\lambda_1, \ldots, \lambda_r$ are numbered in accordance with decreasing order of their real parts. 
That is, we number the eigenvalues such that
      $$\Re\, \lambda_1 \ge \Re\, \lambda_2 \ge \cdots \ge \Re\, \lambda_r.$$  
Then, $\lambda_1$ is called the {\em principal eigenvalue} and the corresponding
idempotent is called the {\em principal idempotent}. 

A random variable uniformly distributed on $[0,1]$ is denoted by Uniform~$[0,1]$.        
\section{Probabilistic analysis}
For the sake of the analysis, we index the interval produced after processing~$n$ bits.
After $n$ bits have been processed  by the coder, the algorithm narrows down the 
interval $[X_{n-1}, Y_{n-1}]$ to one with lower end $X_n$ and upper 
end $Y_n$. Before receiving any bits, 
the initial interval is $[X_0, Y_0]  = [0,1]$. 

We have a recursive formulation. Let $\field_n$ be the sigma field generated by step $n$ (all the information available by time $n$, such as $X_n$ and $Y_n$).
Note that $B_{n+1}$ is independent of $\field_n$, and we have the conditional distributions:
$$X_{n+1} \given \field_n = \begin{cases}
                    X_n,  &\mbox{with probability}\ q;\\
                    X_n + q(Y_n-X_n), &\mbox{with probability}\ p,
          \end{cases}$$
and  
$$Y_{n+1} \given \field_n=  \begin{cases}
                    X_n + q(Y_n-X_n),  &\mbox{with probability}\ q;\\
                    Y_n, &\mbox{with probability}\ p.
          \end{cases}$$
Toward a characterization of the distribution of the random codeword,
we deal with the bivariate moment generating function of the row vector $(X_n, Y_n)$,
which is namely
\begin{align*}
\phi_{X_n, Y_n} (u,v) &:= \E[e^{uX_n + vY_n }] \\
          &= 1 + \E[X_n]u + \E[Y_n] v\\
                  &\qquad + \frac 1 {2!} \big (\E[X_n^2]\, u^2
                        + 2\, \E[X_nY_n]\, uv + \E[Y_n^2]\,v^2\big) + 
                         \cdots,
\end{align*}                  
a machinery to generate all moments. 

By conditioning on $\field_n$, we write the functional equation
\begin{align}
\phi_{X_{n+1},Y_{n+1}} (u,v) 
        &= \E\big[\E[e^{uX_{n+1} + vY_{n+1}}]\, \big | \, \field_n\big] \nonumber \\
        &=  q \, \E[e^{uX_n+v(X_n +q(Y_n-X_n))}]
                       + p \, \E[e^{u(X_n +q(Y_n - X_n))+ vY_n}\big] \nonumber \\ 
        &=  q \, \E[e^{(u+pv)X_n + qvY_n}]
                       + p \, \E[e^{puX_n +(qu+v)Y_n}\big]\nonumber \\ 
         &=  q\, \phi_{X_n,Y_n} (u+pv,qv)   + p\, \phi_{X_n,Y_n} (pu,qu+v).      
         \label{Eq:functional}                    
\end{align}

\section{The moments} 
Toward the computation of the mixed moment $\E[X^i Y^j]$, for $i\ge 0$ and $j\ge 0$,
we take the $m\mbox{th}=(i+j)\mbox{th}$ derivatives of the functional equation~(\ref{Eq:functional}), that is,  $i$ times with respect to $u$, then again $j$ times with respect to $v$. We evaluate all these expressions at $u=v=0$. We get recurrence relations:

\begin{align*}
\E[X_{n+1}^i Y_{n+1}^j] &=  \frac {\partial^{i+j}} {\partial^i u\, \partial^j v}\, \E[e^{uX_{n+1}
                                 + v Y_{n+1}}]_{u=v=0} \\
              &= q\, \frac {\partial^{i+j}} {\partial^i u\, \partial^j v}\, 
         \E[e^{(u+pv)X_n+ qvY_n}]_{u=v=0} \\
            &\qquad {}+  p\, \frac {\partial^{i+j}} {\partial^i u\, \partial^j v}  
                      \, \E[e^{puX_n+ (qu+v)Y_n}]_{u=v=0} \\
            &= q\, \E\big[X_n^i(pX_n+qY_n)^j\big] 
                                 +  p\, \E\big[(pX_n+qY_n)^i  Y_n^j\big]   \\
             &   =q\, \E\Big[\sum_{k=0}^j X_n^i(pX_n)^k (qY_n)^{j-k} {j \choose k} \Big] \\        
             &\qquad\qquad \qquad {}  +  p\, \E\Big[\sum_{\ell=0}^i (pX_n)^\ell 
                             (qY_n)^{i-\ell} {i \choose \ell}Y_n^j\Big] \\
             &=q\, \sum_{k=0}^j p^k q^{j-k} {j \choose k}\E[X_n^{k+i} Y_n^{j-k}] \\
                   &\qquad\qquad \qquad  +  p\, \sum_{\ell=0}^i 
                  p^\ell q^{i-\ell} {i \choose \ell}\E[X_n^\ell 
                             Y_n^{i+j-\ell}].
\end{align*}

It is convenient to put the moments of order $i+j=m$ in a vector to discern a recurrence relation 
in matric form:
\begin{equation}\begin{pmatrix} \E[X_{n+1}^m] \\ 
                              \E[X_{n+1}^{m-1}Y_{n+1}]\\
                              \vdots\\
                              \E[Y_{n+1}^m]  \end{pmatrix} 
         = \matW_m
                                      \begin{pmatrix}  \E[X_n^m] \\ 
                              \E[X_n^{m-1}Y_n]\\
                              \vdots\\
                              \E[Y_n^m] \end{pmatrix} 
                                          = \matW^{n+1}_m\begin{pmatrix}  0 \\ 
                              0\\
                               \vdots\\   
                               0\\   
                              1\end{pmatrix}, 
         \label{Eq:moments}                      
\end{equation}                              
 where, the matrix $\matW_m$ is $(m+1)$-by-$(m+1)$, and the row corresponding to
 $\E[X_{n+1}^i Y_{n+1}^j]$ is
 \begin{align}
 &\Big[{j\choose j}  p^j q \qquad {j\choose j-1}  p^{j-1} q^2  \qquad \ldots \qquad pq^j {j\choose 1} \qquad p^{i+1} + q^{j+1}\nonumber\\
   &\qquad  \qquad \qquad \qquad \quad   p^iq {i \choose 1}\qquad  \ldots \qquad  p^2q^{i-1}  {i\choose i-1} \qquad pq^i  {i\choose i}\Big] .
 \label{Eq:Wrow}  
 \end{align}
\begin{theorem}
For each $m=1,2,\ldots$, the matrix $\matW_m$ is a regular doubly stochastic matrix.
\end{theorem}
\begin{proof}
Adding up the entries in $\matW_m$ across any row (cf.~(\ref{Eq:Wrow})), the binomial theorem gives us
$$    q\, \sum_{k=0}^j p^k q^{j-k} {j \choose k}
                    +  p\, \sum_{\ell=0}^i 
                  p^\ell q^{i-\ell} {i \choose \ell} = p+q = 1.$$
Therefore, $\matW_m$ is a row-stochastic matrix.            

In order to demonstrate that $\matW_m$ is column-stochastic, we prove that the sum of adjacent columns are the same.
Observe that the matrix $\matW_m$ is a sum of an upper triangular matrix $\matU_m$ and a lower triangular matrix
$\matL_m$, where
$$
(\matL_m)_{ij}=
\tbinom{i}{i-j}p^{i-j}q^{j+1},
\quad
\mbox{and}
\quad
(\matU_m)_{ij}=
\tbinom{m-i}{j-i}p^{m+1-j}q^{j-i},\quad 0\le i,  j\le m;
$$
with the binomial coefficients taken, as usual, to be 0, when the lower index is larger than the upper index.

The sum of the $j$th column of $\matL_m$ is
\begin{align}
\label{eq:sumrowLmj}
&\Big(
\tbinom{j}{0}+
\tbinom{j+1}{1}p+
\tbinom{j+2}{2}p
+\cdots+
\tbinom{m}{m-j}p^{m-j}
\Big)q^{1+j}\nonumber\\
&\qquad =
\Big(
\tbinom{j}{0}+
\tbinom{j+1}{1}p+
\tbinom{j+2}{2}p^2
+\cdots+
\tbinom{m}{m-j}p^{m-j}
\Big)q^{j}\nonumber\\
&\qquad \qquad  {} -\Big(
\tbinom{j}{0}p+
\tbinom{j+1}{1}p^2
+\cdots+
\tbinom{m-1}{m-j-1}p^{m-j}
+\tbinom{m}{m-j}p^{m-j+1}
\Big)q^{j},
\end{align}
where we have used that $q^{1+j}=(1-p)q^j$.
The $p^i$ terms are collected via Pascal's identity 
$\tbinom{r}{s}-\tbinom{r-1}{s-1}=\tbinom{r-1}{s}$.
The right-hand side in~(\ref{eq:sumrowLmj}) is therefore
\begin{equation*}
\left(
\tbinom{j-1}{0}+
\tbinom{j}{1}p+
\tbinom{j+1}{2}p^2+
\cdots+
\tbinom{m}{m-j+1}p^{m-j+1}\right)q^{j}
-\tbinom{m+1}{m-j+1}p^{m-j+1}q^{j}.
\end{equation*}
This is the sum of the $(j-1)$st
column of $\matL_m$ minus $\tbinom{m+1}{m-j+1}p^{m-j+1}q^{j}$.

The sum of the $(j-1)$st
column of $\matU_m$ is calculated similarly:
\begin{align*}
&\Big(
 \tbinom{m}{j-1}q^{j-1}
+\tbinom{m-1}{j-2}q^{j-2}
+\cdots
+\tbinom{m+1-j}{0}
\Big)p^{m+2-j}\\
&\qquad\qquad=\Big(
 \tbinom{m}{j-1}q^{j-1}
+\tbinom{m-1}{j-2}q^{j-2}
+\cdots
+\tbinom{m+1-j}{0}
\Big)p^{m+1-j}\\
&\qquad \qquad \qquad \qquad-\Big(
 \tbinom{m}{j-1}q^{j}
+\tbinom{m-1}{j-2}q^{j-1}
+\cdots
+\tbinom{m+1-j}{0}q
\Big)p^{m+1-j}\\
&\qquad \qquad =\Big(
\tbinom{m}{j}q^{j}
+\cdots
+\tbinom{m-j}{0}\Big)p^{m+1-j}
-\tbinom{m+1}{j}q^{j}p^{m+1-j},
\end{align*}
which is the sum of the $j$th
column of $\matU_m$ minus $\tbinom{m+1}{j}q^{j}p^{m+1-j}$.

Consequently, the sum of the $j$th
column of $\matW_m$ is equal to the sum of the $(j-1)$st
column of $\matW_m$. Since $\matW_m$ is a
row-stochastic matrix, the total sum of all its elements equals  
the number of rows, $m+1$. Therefore, the sum of each column is equal to $1$. We conclude that $\matW_m$ is 
column-stochastic too, 
for all $m\ge 1$. The matrix $\matW_m$ is regular since all its entries are positive.
\end{proof}
\begin{lemma}\label{lem:eigenvectormatrix}
A modal matrix $\matP_m$ of\, $\matU_m$ is
$(\matP_m)_{ij}=\tbinom{m-i}{j-i},$ with the inverse $(\matP_m^{-1})_{ij}=(-1)^{i-j}
\tbinom{m-i}{j-i}.$
\end{lemma}
\begin{proof}
We verify that
\begin{align*}
(\matU_m\matP_m)_{ij}&=\sum_{k=i}^{j}
(\matU_m)_{ik}(\matP_m)_{kj}\\
&=\sum_{k=i}^{j}
\tbinom{m-i}{k-i}p^{m+1-k}\,q^{k-i}\,
\tbinom{m-k}{j-k}\\
&=\sum_{k=i}^{j}
\tfrac{(m-i)!}{(k-i)!\,(m-k)!}\, p^{m+1-k}\,q^{k-i}\,
\tfrac{(m-k)!}{(j-k)!\,(m-j)!}\\
&=
\tfrac{(m-i)!}{(j-i)!\,(m-j)!}
\sum_{k=i}^{j}
p^{m+1-k}\,q^{k-i}
\tfrac{(j-i)!}{(j-k)!\,(k-i)!}\\
&=
\tbinom{m-i}{j-i}
\sum_{k=i}^{j}
p^{m+1-k}\,q^{k-i}
\tbinom{j-i}{k-i}\\
&=
\tbinom{m-i}{j-i}
\sum_{\ell=0}^{j-i}
p^{m-i+1-\ell}q^{\ell}
\tbinom{j-i}{\ell}\\
&=
\tbinom{m-i}{j-i}
p^{m-j+1}(p+q)^{j-i}\\
&=p^{m-j+1}(\matP_m)_{ij}.
\end{align*}
We also have
\begin{align*}
(\matP_m\matP_m^{-1})_{ij}
&=\sum_{k=i}^j(\matP_m)_{ik}(\matP_m^{-1})_{kj}\\
&=\sum_{k=i}^j\tbinom{m-i}{k-i}
 \tbinom{m-k}{j-k}(-1)^{j-k}\\
&=\sum_{k=i}^j
\tfrac{(m-i)!}{(k-i)!\,(m-k)!} \times
\tfrac{(m-k)!}{(j-k)!\,(m-j)!}
(-1)^{j-k}\\
&=\tfrac{(m-i)!}{(m-j)!\, (j-i)!}
\sum_{k=i}^j
\tfrac{(j-i)!}{(j-k)!\,(k-i)!}
(-1)^{j-k}\\
&=\tbinom{m-i}{m-j}
\sum_{k=i}^j
\tbinom{j-i}{j-k}
(-1)^{j-k}\\
&=\delta_{ij},
\end{align*}
where  $\delta_{ij}$ is Kronecker's delta function that assumes the value 1 only when its two arguments are the same; otherwise it is 0.
\end{proof}
\begin{lemma}
Let $\matP_m$ be the 
modal matrix of\, $\matU_n$ from Lemma~\ref{lem:eigenvectormatrix}.
Then
$$(\matP_m^{-1}\matL_m\matP_m)_{ij}=\tbinom ijp^{i-j}q^{m+1-i}.$$ 
\end{lemma}
\begin{proof} Let $\matT_m=\matP_m^{-1}\matL_m\matP_m$. It is sufficient to show that $\matP_m\matT_m=\matL_m\matP_m$. We will prove this by induction. The lemma is trivially true for $m=1$. Assume that
$\matP_{m-1}\matT_{m-1}=\matL_{m-1}\matP_{m-1}$.
Observe that{\footnote{In the display, each zero in boldface is a matrix or a vector of zeros
of appropriate size.}}
$$
\matL_m=\left(
\begin{array}{c|c}
q & {\bf 0} \\\hline
{\bf 0} &q\matL_{m-1}
\end{array}
\right)+\left(
\begin{array}{c|c}
{\bf 0}  &  0  \\\hline
p\matL_{m-1}& {\bf 0} 
\end{array}
\right),
$$
$$\matP_m=\left(
\begin{array}{c|c}
1 & \tbinom m1\cdots\tbinom mm \\\hline
{\bf 0} &\matP_{m-1}
\end{array}
\right)=\left(
\begin{array}{c|c}
\matP_{m-1}&{\bf 0} \\\hline
{\bf 0} &1
\end{array}
\right)+\left(
\begin{array}{c|c}
{\bf 0} &\matP_{m-1} \\\hline
0&{\bf 0} 
\end{array}
\right),$$
and
$$\matT_m=\left(\begin{array}{c|c}
    q^{m+1}&{\bf 0}   \\\hline
     {\bf 0} &\matT_{m-1} 
\end{array}\right)
+\left(\begin{array}{c|c}
     {\bf 0} & 0  \\\hline
p\matT_{m-1} &{\bf 0} 
\end{array}\right).
$$
The product of $\matL_m$ with $\matP_m$ is
\begin{align*}
    \matL_m\matP_m&=
q\left(
\begin{array}{c|c}
1&\tbinom m1\cdots\tbinom mm \\\hline
{\bf 0} &\matL_{m-1}\matP_{m-1}
\end{array}
\right)
+
\left(
\begin{array}{c|c}
{\bf 0} &0 \\\hline
p\matL_{m-1}\matP_{m-1}&{\bf 0} 
\end{array}
\right)\\
&\qquad\qquad{}+
\left(
\begin{array}{c|c}
0&{\bf 0} \\\hline
{\bf 0}&p\matL_{m-1}\matP_{m-1}
\end{array}
\right)\\
&=\left(
\begin{array}{c|c}
q&q\tbinom m1\cdots q\tbinom mm \\\hline
{\bf 0}&\matP_{m-1}\matT_{m-1}
\end{array}
\right)+
\left(
\begin{array}{c|c}
{\bf 0}&0 \\\hline
p\matP_{m-1}\matT_{m-1}&{\bf 0}
\end{array}
\right).
\end{align*}
On the other hand, we have
\begin{align*}
    \matP_m\matT_m&=\left(
    \begin{array}{c|c}
        q & q\tbinom m1\cdots q\tbinom mm \\\hline
         {\bf 0} & \matP_{m-1}\matT_{m-1}
    \end{array}
    \right)+
    \left(
    \begin{array}{c|c}
       {\bf 0} & 0 \\\hline
       p\matP_{m-1}\matT_{m-1} &  {\bf 0} 
    \end{array}
    \right).
\end{align*}
We have used that $\left[\tbinom m0,\tbinom m1,\ldots,\tbinom mm\right]$ is a left eigenvector of $\matT_m$ for the eigenvalue $q$. In fact, $\left[\tbinom m0,\tbinom m1,\ldots,\tbinom mm\right]\matP_m^{-1}=[1,0,\ldots,0]$ follows from Lemma~\ref{lem:eigenvectormatrix}. Moreover, observe that $[1,0,\ldots,0]$ is a left eigenvector of $\matL_m$ with respect to its eigenvalue $q$ and finally, $[1,0,\ldots,0]\, \matP_m=\left[\tbinom m0,\ldots,\tbinom mm\right]$. By induction on $m$, we conclude that $\matL_m\matP_m=\matP_m\matT_m$ for all $m\geq 1$.
\end{proof}
\begin{cor}
\label{Cor:eigenvalues}
The eigenvalues of\, $\matW_m$ are 
$1,p^2+q^2,p^3+q^3,\ldots,p^{m+1}+q^{m+1}$.
\end{cor}
\begin{lemma}
\label{Lm:perpendicular}
    Let $\matA$ be an $n$-by-$n$ doubly stochastic  matrix. 
    Then, $\matA$ has
    a simple (multiplicity 1) principal eigenvalue $\lambda_1 =1$, and a principal eigenvector 
    $\mathbf{e}$ of all ones.
   The eigenvectors
    for the eigenvalues $\lambda\not =1$ are perpendicular to $\mathbf{e}$.
\end{lemma}
\begin{proof}
The assertion that the principal eigenvalue is simple and of value 1 follows from  Frobenius-Perron theorem~\cite{Frobenius,Perron}; see also~\cite{Horn}, Section 8.4. Consequently, $\bf e$ is a vector of all ones.

Let $\mathbf{x}$ be an eigenvector for an eigenvalue $\lambda\not =1$. Then, 
$\mathbf{e}^T \matA\mathbf{x}=\lambda\mathbf{e}^T\mathbf{x}$. On the other hand, $\matA$ is doubly stochastic, 
and $\mathbf{e}^T \matA\mathbf{x}=\mathbf{e}^T \mathbf{x}$. Therefore, we have 
$(1-\lambda) \mathbf{e}^T\mathbf{x}=0$. Since $\lambda\neq1$, we conclude that $\mathbf{e}^t\mathbf{x}=0$.
\end{proof}
\begin{cor}
 \label{Cor:idempotents}     
    The principal idempotent of $\matW_m$ is the matrix $\frac 1{m+1}{\bf J}_m$, where ${\bf J}_m$ is the $(m+1)$-by-$(m+1)$ matrix with all elements equal to $1$.
\end{cor}
The strategy toward an exact or asymptotic representation is to evaluate $\matW_m^n$ via its
eigenvalues and eigenvectors.
For moments of low order (small $m$), such as the first and second, 
the matrix $\matW_m$ is of low dimensions and we can find exact forms. For higher 
moments (larger $m$), the exact computation of the eigenvalues and idempotents
is forbidding in view of the known combinatorial explosion. For higher moments we cut through
computation by focusing only on the asymptotics.                                            
\subsection{Exact and asymptotic means}   
For $m=1$, 
the matrix $\matW_1$ is a simple 2-by-2 matrix and we have the exact explicit form 
$$\begin{pmatrix} 
         \E[X_n ] \\ 
         \E[Y_n]  
    \end{pmatrix} 
         = \begin{pmatrix} p^2+q &pq \\
                             pq &p+q^2\end{pmatrix}^n   \begin{pmatrix} 0\\ 1 \end{pmatrix}. $$

To get an explicit form, we represent the solution in terms of powers of the eigenvalues and the idempotents of $\matW_1$ to get
$$\matW_1^n 
    = \lambda_1^n {\cal E}_1 + \lambda_2^n {\cal E}_2.$$ 
By Corollary~\ref{Cor:eigenvalues}, the two eigenvalues of $\matW_1$ are
$$\lambda_1 = 1, \qquad \lambda_2 = 1- 2pq = p^2 +q^2\le 1;$$
and from Lemma~\ref{Lm:perpendicular} and Corollary~\ref{Cor:idempotents}, 
the corresponding idempotents are
$${\cal E}_1 = \frac 1 2 \begin{pmatrix}1 &1\\
                                               1 &1 
                      \end{pmatrix} , \qquad {\cal E}_2 = \frac 1 2 
                      \begin{pmatrix}1 &-1\\
                                               -1 &1\end{pmatrix}.  $$                                          
It is quite remarkable that the idempotents are independent of $p$.  
Using these calculations, we get
\begin{align*}
\begin{pmatrix} \E[X_n] \\ \E[Y_n]  \end{pmatrix} 
    &= \bigg(1^n \times \frac 1 2    \begin{pmatrix}1 &1\\
                                               1 &1 
                      \end{pmatrix} +\frac 12  (1-2pq)^n  \begin{pmatrix}
                                        1 &-1\\
                                               -1 &1
                       \end{pmatrix}\bigg)                \begin{pmatrix} 0\\ 1 \end{pmatrix}.     \\            
    &= \frac 1 2 \begin{pmatrix}  
             1 -  (1-2pq)^n\\
             1 + (1-2pq)^n \end{pmatrix}.                 
\end{align*}          

As $n\to \infty$, the codeword is any number in the  interval 
$$\frac 1 2 \big(1  - (1-2pq)^n , 1 +( 1
       - 2pq)^n \big).$$
In implementations that take the middle point
the average coded message is $0.5$.
It is intriguing that the average value of the codeword is the same for all $p\in (0,1)$. 

However, $p$ governs the rate of convergence as we argue next. 
The length of the interval delivered by the arithmetic coder is $L(n) =  (2p^2 - 2p + 1)^n$. 
The function $2p^2 - 2p + 1$ is parabolic and symmetric around  $p=\frct12$ and opens upwards.
So, the average length of the interval is $\frct 1 2^n$ at $p = \frct 1 2$.  As~$p$ moves away from $\frct 1 2$, the average interval gets wider, but still its length is exponentially small. For example, at $p= \frct 1 3$, the length of the interval is $(\frct 5 9)^n$. 

The exponential smallness of the average interval length produced by the 
arithmetic coder requires computations of very small numbers pushing the algorithm into underflow errors. These errors have been reported as a drawback of the performance of arithmetic coding, and remedies have been proposed in~\cite{Chang}.
\subsection{The second moments and concentration laws}      
For $m=2$, we have  $$\matW_2 := \begin{pmatrix} p^3+q &2p^2q &pq^2\\
                                      pq &p^2+q^2&pq\\
                                      p^2q&2pq^2&p+q^3
                                      \end{pmatrix}.  $$  
                                      
By Corollary~\ref{Cor:eigenvalues}                    
the matrix $\matW_2$ has the three eigenvalues
$$\lambda_1 = 1, \qquad \lambda_2 = 1 - 2pq ,\qquad \lambda_3 =   1 - 3pq .$$
Note that $\lambda_2$ and $\lambda_3$ are both real-valued with magnitude less than 1.
Hence, the dominant asymptotics in the expression\footnote{The idempotents ${\cal E}_1$
and  ${\cal E}_2$ of $\matW_2$ different for various $m$. We are only reusing the notation to avoid double indexing.}  
$$\begin{pmatrix} \E[X_n^2] \\ 
                              \E[X_nY_n]\\
                              \E[Y_n^2]  \end{pmatrix} 
    = \lambda_1^n {\cal E}_1 + \lambda_2^n {\cal E}_2 +\lambda_3^n {\cal E}_3 .$$ 
come from $\lambda_1^n$ and its associated idempotent (see 
Corollary~\ref{Cor:idempotents}):
$${\cal E}_1 = \frac 1 3
                    \begin{pmatrix} 1 & 1&1 \\ 
                                              1 & 1&1 \\ 
                                              1 & 1&1 \\  \end{pmatrix}.$$ 
As $n \to 1$, we get\footnote{We use the notation ${\bf  \Theta} (g(n))$ for a vector of the appropriate dimensions (determined by the context), with all its entries being $ \Theta(g(n))$ in the usual sense
of exact order of growth: One says that a function $h(n)$ is of exact order $g(n)$, written $h(n) =  \Theta(g(n))$, if there are two positive real numbers $C_1$ and $C_2$ and a 
natural number~$n_0$, such that $C_1 |g(n)| \le |h(n)|\le C_2|g(n)|$, for all $n\ge n_0$.}
$$\begin{pmatrix}
  \E[X_n^2] \\ 
  \E[X_nY_n]\\
  \E[Y_n^2]
\end{pmatrix}
= \lambda_1 ^n {\cal E}_1
\begin{pmatrix}
  0\\
  0\\
  1
\end{pmatrix}
+ {\bf  \Theta} (\lambda_2^n) = \frac 1 3
\begin{pmatrix}
  1 \\
  1 \\
  1
\end{pmatrix}
+ {\bf \Theta} (\lambda_2^n). $$
We thus obtain the asymptotic variance-covariance matrix
$$\frac 1{12 }\begin{pmatrix}
                              1 &1\\ 
                              1 &1
              \end{pmatrix} +   {\bf  \Theta} (\lambda_2^n).$$ 
Note again that the limiting covariance matrix is independent of $p$.                                                       
The exponentially small rate of convergence in the variances and the covariance 
are suggestive of a strong concentration around the mean. Indeed, we have
\begin{align*}
\V[L_n] &= \V[Y_n -X_n] \\
            &= \V[Y_n] + \V[X_n] -2\, \Cov[X_n, Y_n]\\
            &= \Big(\frac 1 {12}  +   \Theta (\lambda_2^n)\Big)  
              +\Big(\frac 1 {12}  
             +    \Theta (\lambda_2^n)\Big)  -  2\Big(\frac 1 {12}  +    \Theta (\lambda_2^n)\Big)\\
            &=   \Theta (\lambda_2^n).   
\end{align*}
With $\lambda_2 < 1$, Chebyshev's inequality gives (for all $\varepsilon>0)$
$$\Prob(L_n > \varepsilon) \le \frac {\V[L_n]} {\varepsilon^2} =  \Theta (\lambda_2^n) \to 0, \qquad \mbox{as\ } n \to\infty,$$
and we have $L_n\,\inprob\, 0$. 

The exponentially fast rate of convergence allows us to strengthen the result. Summing
all the probabilities, for any $\varepsilon > 0$, we get
$$\sum_{n=0}^\infty\Prob(L_n > \varepsilon) \le  \sum_{n=0}^\infty
             \Theta (\lambda_2^n) < \infty.$$
By the first Borel-Cantelli lemma~\cite{Billingsley}, we can assert that, for any $\varepsilon>0$, we have
 $$\sum_{n=0}^\infty\Prob(L_n > \varepsilon \ \ \mbox{infinitely often}) =0,$$
 which is an equivalent characterization of the convergence $L_n\almostsure 0$. 
\section{The limiting distribution of the codeword} 
We have enough information on the asymptotic moments to execute the method of moments toward
specifying a limit distribution. 
\begin{lemma}
\label{Lm:Wmn}
For any $m\ge 0$, the matrix $\matW_m^n$ has the asymptotic representation
$$\matW_m^n = \frac 1 {m+1}\, {\bf J}_m + {\bf  \Theta}\big((1-2pq)^n\big),$$
where ${\bf J}_m$ is an $(m+1)$-by-$(m+1)$ matrix of all ones.
\end{lemma} 
\begin{proof}
In the expansion into idempotents, we have\footnote{The idempotents ${\cal E}_1$,
${\cal E}_2$
and  ${\cal E}_3$ of $\matW_m$ are not the same as those that appear in  $\matW_1$ or $\matW_2$. 
We are only reusing the notation to avoid double indexing.}  
$$\matW_m^n = \lambda_1^n { \cal E}_1 + \cdots + \lambda_m ^n {\cal E}_{m+1}= 1^n 
         { \cal E }_1 + {\bf  \Theta} \bigl((\lambda_2)^n\bigr).$$  
  The principal idempotent ${\cal E}_1$ is $\frac1 {m+1}\, {\bf J}_m$.        
\end{proof}    
\begin{theorem}
The arithmetic coding algorithm produces a number that converges almost surely to  a {\rm Uniform} $[0,1]$ random variable.
\end{theorem}
\begin{proof}   
From the representation~(\ref{Eq:moments}) and Lemma~\ref{Lm:Wmn}, 
we have\footnote{The asymptotics here are taken component-wise.}
$$\begin{pmatrix} \E[X_n^m] \\ 
                              \E[X_n^{m-1}Y_n]\\
                              \vdots\\
                              \E[Y_n^m]  \end{pmatrix} 
                                          = \matW^n_m\begin{pmatrix}  0 \\ 
                              0\\
                               \vdots\\   
                               0\\   
                              1\end{pmatrix} \to
                              \frac 1 {m+1} {\bf J}_m\begin{pmatrix}  0 \\ 
                              0\\
                               \vdots\\   
                               0\\   
                              1\end{pmatrix} = 
                              \frac 1 {m+1}\begin{pmatrix}  1 \\ 
                              1\\
                               \vdots\\   
                              1\end{pmatrix}.  $$
The asymptotic $m$th moment of $X_n$ is $\frac 1 {m+1}$. 
These are the moments of the standard uniform random variable, Uniform [0,1]. 
Since the uniform distribution is uniquely determined by its moments, 
the method of moments asserts that~$X_n$ converges in distribution
to the Uniform [0,1] random variable. However, it was demonstrated that 
$X_n$ converges almost surely to a limit. That limiting random variable must have a 
Uniform $[0,1]$ distribution, since convergence almost surely implies convergence in distribution.
\end{proof} 
\section{Concluding remarks}
For a class of Bernoulli ($p$) independent input sequences, we investigated the performance of arithmetic coding using tools from probability theory
and linear algebra. We determined the distribution of  the almost-sure limit of  
each end of
the interval delivered by the algorithm, as well as the distribution of the almost-sure limit of the codeword itself.  
These almost-sure distributions are Uniform~[0,1], specified with with rates of convergence. 
Curiously, these limits do not depend on $p$, showing intrinsic robustness of arithmetic coding.
The content of this paper bolsters and upgrades the implicit mention of uniform distributions associated with arithmetic coding in Section 13.1 of~\cite{Cover}.
\section{Acknowledgment}
We thank Jenny M.\ Rivertz for proofreading the manuscript.
          
\end{document}